\documentclass[12pt]{article}
\usepackage[left=1in,right=1in,top=1in,bottom=1in]{geometry}
\usepackage{amssymb,amsthm,amsmath}
\usepackage[hidelinks]{hyperref}
\usepackage{natbib}

\usepackage[affil-it]{authblk}

\newtheorem{thm}{Theorem}
\newtheorem{prop}{Proposition}
\newtheorem{cor}{Corollary}
\theoremstyle{definition}
\newtheorem{definition}{Definition}

\title{On Noncontextual, Non-Kolmogorovian Hidden Variable Theories\thanks{SCF would like to thank audiences at Budapest, T\"ubingen, and Saig, Germany, especially Guido Bacciagalupi and Fay Dowker, for their comments.  BHF would like to thank audiences at the Perimeter Institute and Chapman University.  Both authors acknowledge the support of National Science Foundation Graduate Research Fellowships.}}
\author{Benjamin H.~Feintzeig}
\affil{Department of Philosophy\\University of Washington}
\author{Samuel C.~Fletcher}
\affil{Department of Philosophy\\University of Minnesota\\\&\\Munich Center for Mathematical Philosophy\\Ludwig Maximilian University}

\begin{document}
\maketitle
\begin{abstract}
One implication of Bell's theorem is that there cannot in general be hidden variable models for quantum mechanics that both are noncontextual and retain the structure of a classical probability space. 
Thus, some hidden variable programs aim to retain noncontextuality at the cost of using a generalization of the Kolmogorov probability axioms.  
We generalize a theorem of \citet{Fe15} to show that such programs are committed to the existence of a finite null cover for some quantum mechanical experiments, i.e., a finite collection of probability zero events whose disjunction exhausts the space of experimental possibilities.
\end{abstract}

\section{Introduction}
Bell's theorem \citep{Be64} and its generalizations (e.g., by \citet{ClHoShHo69,Be71,ClHo74,As83}; and \citet{Me86}) restrict which kinds of hidden variable theories for quantum mechanics are consistent with empirical observation and other theoretical commitments.  
In particular, a version due to \citet{Fi82a,Fi82b} and later presented geometrically by \citet{Pi89} shows that the statistics of a quantum mechanical experiment satisfy all Bell-type inequalities if and only if it may be given a classical (Kolmogorovian) probability space hidden variable model satisfying very weak constraints.  
The empirical violation of these inequalities in experiments \citep{AsDaRo82,AsGrRo82, Gietal13, Gietal15}, as predicted by quantum mechanics, rules against such classical probability models. 
Thus, some researchers still interested in pursuing hidden variable models, or more generally in giving realist probabilistic interpretations of quantum theory, have resorted to generalizations of classical probability theory.

They have not, however, always considered the consistency of such alternative theories with another no-go result, the Kochen-Specker (KS) theorem \citep{KoSp67}, which rules out any hidden variable theory that assigns noncontextual, definite values to all physical quantities and preserves standard functional (or logical) relationships between them.
(See also \citet{Be66}.)
Because many of these theories are committed to noncontextuality and definite values, the KS theorem would seem to bear upon them.
Generalizing a theorem of \citet{Fe15}, we show in \S\ref{sec:theorem} that the KS theorem in fact implies that a wide variety of noncontextual hidden variable theories are committed to the existence of a finite null cover for some quantum mechanical experiments.
That is, they are committed, for some choices of possible observables, to a finite collection of measurable events that cover the hidden variable sample space but are each assigned measure zero. 
In simple terms, they are committed to a finite collection of probability zero events whose disjunction exhausts the space of possibilities.

We apply this result in \S\ref{sec:applications} to a variety of models, theories, and approaches that employ generalizations of classical probability spaces:
\begin{itemize}
\item Generalized probability spaces (\S\ref{sec:GPS}), which relax the algebra of classical probability theory to an additive class.\footnote{Kolmogorovian probability theory typically assumes a stronger $\sigma$-algebra, which requires the probability measure to be $\sigma$-additive, that is, additive over countable collections of (disjoint) sets, rather than just finite ones. (Cf. definition \ref{def:CPS}.)  Our results, however, nowhere appeal to this stronger condition, so we have dropped it everywhere for clarity of expression.}
\item Nonmonotonic/negative/extended probability spaces (\S\ref{sec:EPS}), including $p$-adic versions, which allow the probability measure to take on negative values, and the more general complex probability spaces, which allow it to be valued in $\mathbb{C}$.  The former include ``single real history'' versions of the decoherent histories formalism.
\item Upper (lower) probability spaces (\S\ref{sec:U(L)PS}), which relax the additivity of the probability measure on disjoint sets to sub- (resp. super-)additivity.
\item Quantum measure theory (\S\ref{sec:QMT}), which relaxes additivity of the probability measure on pairs of disjoint sets to additivity of triplets.
\end{itemize}
Yet other approaches may also fall under its scope.

In our discussion in \S\ref{sec:FNC}, we show how any such theory must commit its users to a sort of incoherent belief if it is to recover the predictions of quantum mechanics for KS experiments.
Thus it is not even clear that an agent with complete knowledge of the hidden states would even be able to use that knowledge as a guide to rational action.
These features cast doubt on the idea that these generalizations capture any notion of coherent probability at all, despite these spaces' formal similarities with classical probability spaces.  
Consequently, it is unclear whether they have any advantages outweighing their disadvantages over contextual theories in their ability to explain puzzling features of quantum theory.

\section{The Central Theorem}
\label{sec:theorem}

To motivate our theorem, we begin with a certain perspective on Bell's theorem, its generalizations, and hidden variable theories.  Our perspective focuses on structured representations of possible experimental outcomes and their underlying probability theory.

\begin{definition}
	A \textit{simple quantum mechanical experiment} is a triple
	$(\mathcal{H},\psi,\mathcal{O}_n)$, where $\mathcal{H}$ is a Hilbert space, $\psi \in \mathcal{H}$ is a unit vector, and $\mathcal{O}_n = \{ P_1, \ldots, P_n \}$ is a set of $n$ projection operators on $\mathcal{H}$.
\end{definition}

The collection of operators $\mathcal{O}_n$ of a simple quantum mechanical experiment represents the possible projective measurements that can be made on $\psi$.   Each operator $P_i$ corresponds to a ``yes-no" measurement---that is, one with only two possible outcomes---with the probability of a ``yes" outcome given by $\langle \psi |P_i|\psi\rangle$, as usual.
Such experiments are the subject of Bell-type theorems; to state the one of most relevance here, we must introduce the terminology of classical probability theory.\footnote{See \citet{Ma12} and \citet[\S2.2]{Fe15} for more on the relationship of this presentation with more standard presentations of Bell's theorem.}

\begin{definition}
	An \textit{algebra} $\Sigma$ for a nonempty set $X$ is a nonempty subset of the power set $\mathcal{P}(X)$ such that
	\begin{enumerate}
		\item for all $A \in \Sigma$, $X \backslash A \in \Sigma$, and
		\item for all $A,B \in \Sigma$, $A \cup B \in \Sigma$.
	\end{enumerate}
\end{definition}
\noindent Note that if $\Sigma$ is an algebra for $X$, then $\{ \emptyset, X \} \subseteq \Sigma$.

\begin{definition}
	A \textit{classical probability space} is a triple $(X,\Sigma,\mu)$, where $\Sigma$ is an algebra for the nonempty set $X$ and $\mu: \Sigma \to [0,1]$ is such that
	\begin{enumerate}
		\item $\mu(X) = 1$, and
		\item for all disjoint $A,B \in \Sigma$, $\mu(A \cup B) = \mu(A) + \mu(B)$.
	\end{enumerate}
\label{def:CPS}
\end{definition}

The function $\mu$ is typically called the space's \textit{probability measure}, and the elements of $\Sigma$ its \textit{measurable sets}.
Note as well that the second condition, additivity on disjoint pairs, immediately implies additivity on finite collections of mutually disjoint sets.

\begin{definition}
	A \textit{(restricted) classical probability space representation}\footnote{It is ``restricted'' since one might demand the second condition hold for all collections of compatible observables, not just pairs.  Cf. definitions 5 and 6 of \citet{Fe15}.} for a simple quantum mechanical experiment $(\mathcal{H},\psi,\mathcal{O}_n)$ is a classical probability space $(X,\Sigma,\mu)$ and a map $E: \mathcal{O}_n \to \Sigma$ satisfying both of the following conditions:
	\begin{enumerate}
		\item for each $P_i \in \mathcal{O}_n$, $\mu(E(P_i)) = \langle \psi | P_i | \psi \rangle$, and
		\item for each $P_i,P_j \in \mathcal{O}_n$, if $[P_i, P_j] = 0$, then $\mu(E(P_i) \cap E(P_j)) = \langle \psi | P_i P_j | \psi \rangle$.
	\end{enumerate}
\end{definition}

Classical probability space representations for simple quantum mechanical experiments thus have a measurable set corresponding to each possible measurement outcome, with the measure of the former equal to the probability value of the latter.  
If two possible measurements are compatible, (i.e., if they commute,) then the probability value of a ``yes" outcome on their joint measurement---i.e., a ``yes" outcome on both measurements---is the measure assigned to the intersection of their corresponding measurable sets.

There is a close connection between classical probability space representations and Bell-type inequalities:

\begin{thm}[\citealt{Fi82a,Fi82b,Pi89}]
\label{thm:Pitowsky}
	The statistics for outcomes of a simple quantum mechanical experiment satisfy all Bell-type inequalities if and only if the experiment has a classical probability space representation.
\end{thm}

It is well known that the statistics for some simple quantum mechanical experiments, namely the EPR-Bohm setup, violate Bell-type inequalities, so it follows from theorem \ref{thm:Pitowsky} that those experiments cannot be given a classical probability space representation.
Hidden variable theories, though, tend to have the ambition to provide a coherent and consistent representation for \textit{all} quantum mechanical experiments.
Thus theorem \ref{thm:Pitowsky} can be thought of as a no-go result for the use of classical probability spaces as hidden variable theories.  For ease of comparison with the no-go results that we present later, we will state this first no-go result explicitly as follows.

\begin{cor}
\label{cor:Pitowsky}
There exist simple quantum mechanical experiments with no classical probability space representation.
\end{cor}

Consider, then, the following much broader type of representation.

\begin{definition}
	A \textit{weak (noncontextual) probability space} is an ordered triple $(X,\Sigma,\mu)$, where $X$ is a nonempty set, $\Sigma\subseteq\mathcal{P}(X)$, and $\mu:\Sigma\to Y\supseteq[0,1]$.
\end{definition}

\begin{definition}
	A \textit{weak (noncontextual) hidden variable representation} for a simple quantum mechanical experiment $(\mathcal{H},\psi,\mathcal{O}_n)$ is a quadruple $(X,\Sigma,\mu,E)$, consisting of a weak probability space $(X,\Sigma,\mu)$ and a map $E: \mathcal{O}_n \to \Sigma$ satisfying both of the following conditions:
	\begin{enumerate}
		\item for each $P_i \in \mathcal{O}_n$, $\mu(E(P_i)) = \langle \psi | P_i | \psi \rangle$, and
		\item for each $P_i,P_j \in \mathcal{O}_n$, if $P_i \perp P_j$, then $E(P_i) \cap E(P_j) \in \Sigma$ and $\mu(E(P_i) \cap E(P_j)) = \langle \psi | P_i P_j | \psi \rangle = 0$.
	\end{enumerate}
\end{definition}

Weak hidden variable representations demand much less than classical probability space representations.  
First, weak representations do not require the event space $\Sigma$ to be an algebra for $X$ or the function $\mu$ to be a probability measure. The event space $\Sigma$ need not have any nice algebraic properties at all, and the function $\mu$ can be valued in any set containing the real unit interval, including complex or negative numbers.
Second, while weak representations still require each possible measurement of the experiment to have a corresponding ``measurable'' set, instead of requiring \textit{all} possible joint measurements to have as corresponding ``measurable'' sets the intersection of those for their single component measurements, they only require this for orthogonal observables.
(Note that this is a strict subclass, since orthogonal observables are commuting observables but not vice versa.)
The concept of a weak hidden variable representation is not intended by itself as a model for a hidden variable theory.
Indeed, actual proposed noncontextual hidden variable theories (see \S\ref{sec:applications}) generally add more structure than is demanded by a weak hidden variable representation.
Rather, it is intended as a broad framework that encompasses many actual, more interesting hidden variable proposals as specializations.
Because our theorem concerns weak hidden variable representations, it bears upon all such specializations, as we discuss in section $\S3$.

The engine behind the central theorem is the Kochen-Specker theorem:
\begin{thm}[\citealt{KoSp67}]
\label{thm:KS}
	For any Hilbert space $\mathcal{H}$ with $dim(\mathcal{H})\geq 3$, there is a finite collection of projection operators $\mathcal{O}_n = \{P_1,...,P_n\}$ on it such that there is no function $f : \mathcal{O}_n \to \{0,1\}$ that assigns 1 to exactly one element of every subset of $\mathcal{O}_n$ whose elements are mutually orthogonal and span $\mathcal{H}$.
\end{thm}

In particular, because this theorem can be understood as a proof of existence---that there exists, in particular, a simple quantum mechanical experiment with such and such properties---we shall employ the concept of a witness to the theorem:

\begin{definition}
	A \textit{KS witness} is a simple quantum mechanical experiment $(\mathcal{H},\psi,\mathcal{O}_n)$ such that $\dim(\mathcal{H}) \geq 3$ and there is no function $f : \mathcal{O}_n \to \{0,1\}$ that assigns 1 to exactly one element of every subset of $\mathcal{O}_n$ whose elements are mutually orthogonal and span $\mathcal{H}$.\footnote{Notice that if $(\mathcal{H},\psi,\mathcal{O}_n)$ is a KS witness, then so is $(\mathcal{H},\psi',\mathcal{O}_n)$ for any distinct $\psi'\in\mathcal{H}$ because the KS theorem does not depend on any choice of state vector.  We nevertheless include an arbitrary state vector $\psi$ in all of our results to highlight the structural similarity with Bell-type no-go theorems.}
\end{definition}

We will need only the following two very weak constraints on weak hidden variable representations $(X,\Sigma,\mu,E)$ of a simple quantum mechanical experiment $(\mathcal{H},\psi,
\mathcal{O}_n)$.

\begin{description}
		\item[Weak Classicality (WC)] If $Q \subseteq \mathcal{O}_n$ consists of only mutually orthogonal projection operators spanning $\mathcal{H}$, then $(X, \Sigma_Q, \mu_{|\Sigma_Q})$ is a classical probability space, where $\Sigma_Q$ is the smallest algebra for $X$ containing $\{E(P_i): P_i \in Q\}$.
		\item[No Finite Null Cover ($\neg$FNC)] There is no finite collection $\{B_1, \ldots, B_m\} \subseteq \Sigma$ such that $\mu(B_i)=0$ for all $i \in \{1, \ldots, m\}$ and $\bigcup_{i=1}^m B_i = X$.
	\end{description}
    
WC demands that the statistics associated with special collections of commuting operators, namely ones that are often interpreted as being simultaneously measurable, can always be modeled by a classical probability space.  
This fits with our ordinary empirical practice in which the results for a single setting of an experiment are always indicative of classical probabilities.  
$\neg$FNC, on the other hand, rules out a certain kind of pathology.  
It requires that there is never a finite collection of probability zero events that cover the entire space, or in other words, that there is no finite collection of probability zero events whose disjunction is guaranteed to occur.  Of course, classical probability space representations always satisfy both of these conditions; weak hidden variable representations need not, which is why we must include both WC and $\neg$FNC as explicit additional assumptions.
Using these conditions, we now state our main result.

\begin{thm}
\label{thm:main}
	No KS witness $(\mathcal{H},\psi,\mathcal{O}_n)$ has a weak hidden variable representation $(X,\Sigma,\mu,E)$ satisfying both WC and $\neg$FNC.
\end{thm}

The strategy of the proof is to consider the corresponding measurable sets $E(P_i)$ for all collections of mutually orthogonal projection operators $P_i$ in $\mathcal{O}_n$ that span $\mathcal{H}$.  
One then removes points of the sample space $X$ that are either in none of the sets $E(P_i)$, or are in $E(P_i) \cap E(P_j)$ for orthogonal $P_i$ and $P_j$: these are, respectively, hidden states that realize none of the observables corresponding to the $P_i$ spanning $\mathcal{H}$, and hidden states that realize incompatible (orthogonal) observables.
Hidden states in these two classes \textit{allow} for the construction of what would otherwise be a KS-prohibited function $f: \mathcal{O}_n \to \{0,1\}$ because they break the correspondence between mutually orthogonal operators spanning $\mathcal{H}$ and value definiteness.  
$\neg$FNC guarantees that this procedure does not remove all the points of $X$, and from at least one remaining point, one can construct a KS-prohibited function $f$, leading to a contradiction.

\begin{proof}
For the sake of deriving a contradiction, suppose there is a KS witness $(\mathcal{H},\psi,\mathcal{O}_n)$ with a weak hidden variable representation $(X,\Sigma,\mu,E)$ satisfying WC and $\neg$FNC.
Let
\[C_1 = \{ E(P_i) \cap E(P_j) \in \Sigma : P_i \perp P_j \}\]
and let $R_1 = \bigcup C_1$.  According to this definition, $R_1$ is the collection of all hidden states in $X$ that would yield ``yes" outcomes to some pair of ``contradictory" measurements, i.e., orthogonal projectors.
Further, for any subcollection $Q\subseteq\mathcal{O}_n$, let $T_Q = \bigcup_{P_i \in Q} E(P_i)$, and then define
\[C_2 = \{X \backslash T_Q : Q \subseteq \mathcal{O}_n \text{ contains mutually orthogonal projection operators spanning } \mathcal{H} \}\]
and $R_2 = \bigcup C_2$.  According to this definition, $R_2$ is the collection of all hidden states in $X$ that would yield a ``no" outcome for some collection of ``exhaustive" measurements, i.e. orthogonal projectors spanning $\mathcal{H}$.
Finally, define the ``reduced'' space $X' = X \backslash (R_1 \cup R_2)$ and subsets $A'_i = E(P_i) \cap X'$.

We claim that $X'$ is non-empty.  
To see this, note that for any $Q \subseteq \mathcal{O}_n$ whose members are mutually orthogonal and span $\mathcal{H}$, by WC $(X,\Sigma_Q,\mu_{|\Sigma_Q})$ is a classical probability space with $\Sigma_Q$ the smallest algebra containing $\{E(P_i) : P_i \in Q\}$.  
Since $\Sigma_Q$ is closed under unions, $T_Q \in \Sigma_Q$, hence $\mu(T_Q) = \mu\left(\bigcup_{P_i \in Q} E(P_i)\right) = \sum_{P_i \in Q} \mu(E(P_i)) = \sum_{P_i \in Q} \langle \psi | P_i | \psi \rangle = 1$.  
Since $\Sigma_Q$ is closed under complementation, $X \backslash T_Q \in \Sigma_Q$, hence $\mu(X \backslash T_Q) + \mu(T_Q) = \mu((X \backslash T_Q) \cup T_Q) = \mu(X) = 1$.  
It follows that $\mu(X \backslash T_Q) + 1 = 1$, so $\mu(X \backslash T_Q) = 0$.

Therefore $R_2$ is a finite union of null sets.  
By the hypothesis of a weak hidden variable representation, $\mu(E(P_i) \cap E(P_j)) = 0$ when $P_i \perp P_j$, so $R_1$ is a finite union of null sets, too.  
It then follows from $\neg$FNC that $R_1 \cup R_2$ cannot cover $X$, so $X'$ must be non-empty.

Thus we may consider some $x \in X'$ and define $f_x: \mathcal{O}_n \to \{0,1\}$ as
\begin{equation*}
f_x(P_i) = \begin{cases}
1, \text{ if } x \in A'_i, \\
0, \text{ if } x \notin A'_i,
\end{cases}
\end{equation*}
which is just $\chi_{A'_i}(x)$, the characteristic function of $A'_i$.
Now consider any $Q \subseteq \mathcal{O}_n$ whose members are mutually orthogonal and span $\mathcal{H}$.

First, we claim that $x \in A'_i = E(P_i) \cap X'$ for at least one $P_i \in Q$.  
For, since $X \backslash T_Q \subseteq R_2$ and $X' \cap R_2 = \emptyset$, $X' \cap (X \backslash T_Q) = \emptyset$.  
But since $X' \subseteq X$, $X' \cap (X \backslash T_Q) = X' \backslash (X' \cap T_Q)$, hence $X' \cap T_Q = X'$.  
Thus $\bigcup_{P_i \in Q} A'_i = \bigcup_{P_i \in Q} (E(P_i) \cap X') = X' \cap \left( \bigcup_{P_i \in Q} E(P_i) \right) = X' \cap T_Q = X'$.  
Hence the $A'_i$ cover $X'$, so $x$ must be in at least one $A'_i$.

Second, we claim that $x \in A'_i = E(P_i) \cap X'$ for at most one $P_i \in Q$.  
Note that $A'_i \cap A'_j = (E(P_i) \cap X') \cap (E(P_j) \cap X') = (E(P_i) \cap E(P_j)) \cap X'$. 
Since $P_i \perp P_j$ if $i \neq j$, $E(P_i) \cap E(P_j) \subseteq R_1$.  
Furthermore, $R_1 \cap X' = \emptyset$, so $(E(P_i) \cap E(P_j)) \cap X' = \emptyset$.  
Hence $A'_i \cap A'_j = \emptyset$, which shows that $x$ cannot belong to both $A_i'$ and $A_j'$ for any two distinct $P_i,P_j\in Q$.

Thus, $x \in A'_i$ for exactly one $i$, from which it follows that $f(P_i) = 1$ for exactly one $P_i \in Q$.
Furthermore, this holds for any $Q \subseteq \mathcal{O}_n$ whose members are mutually orthogonal and span $\mathcal{H}$.  
But by theorem \ref{thm:KS}, this function $f$ cannot exist, which yields our desired contradiction.
\end{proof}

This shows that under very weak constraints, one arrives at a no-go result for weak hidden variable representations quite similar in form to the no-go result for classical probability space representations.  It is worth making explicit how the logical structure of this result replicates that of the Bell-type theorems presented above; since the KS theorem guarantees the existence of KS witnesses, we can state the result as follows.

\begin{cor}
\label{cor:main}
There exist simple quantum mechanical experiments with no weak hidden variable representation satisfying both WC and $\neg$FNC.
\end{cor}

If one wants to use some kind of weak hidden variable representation to recover the predictions of quantum mechanics, then one must reject one of the assumptions, WC or $\neg$FNC.  
Either of these options comes with a high price, as we will discuss in \S\ref{sec:FNC}.  
Before doing so, we show that theorem \ref{thm:main} applies to many kinds of hidden variable theories and alternative probability theories that have been explicitly proposed for quantum mechanics.

\section{Applications}
\label{sec:applications}

In this section we apply theorem \ref{thm:main} to a variety of alternative noncontextual probability theories that have been proposed for quantum mechanics.  
Each of these probability theories weakens the assumptions of classical probability spaces in some particular way.  
The generalized probability spaces considered in \S\ref{sec:GPS} relax the requirement that the events form an algebra. 
The extended probability spaces considered in \S\ref{sec:EPS} relax the requirement that the range of the measure is contained in the real unit interval. 
Lastly, both the upper (lower) probability spaces in \S\ref{sec:U(L)PS} and the quantum measure spaces of \S\ref{sec:QMT} relax the second Kolmogorov axiom (additivity), although in slightly different ways.  
We will save discussion of the interpretation of finite null covers for \S\ref{sec:FNC}.

\subsection{Generalized Probability Spaces}
\label{sec:GPS}

Arthur Fine \citeyearpar{Fi82a,Fi82b} has argued that the real lesson we learn from Bell's theorem is that one cannot assign joint probabilities to incompatible observables.  
One way to understand this is as a proposal to use so-called generalized probability spaces, as suggested by \citet{Su65,Su66,KrLuSuTv71}; and \citet{Gu88}.\footnote{Note that this is not the solution Fine favors; rather, he suggests the use of what he calls \textit{prism models} \citep{Fi82c}, which we will not analyze in this paper.}
The complaint against classical probability space representations is that, even though they make no requirements on the actual probability values assigned to conjunctions of incompatible observables, these representations do require that we assign \textit{some probability or other} to such conjunctions.  
This is because if $\Sigma$ is an algebra, then whenever $A,B\in\Sigma$, it follows that $A\cap B\in\Sigma$.  
Generalized probability spaces relax the requirement that the event space $\Sigma$ forms an algebra, allowing us to assign probabilities to the individual events corresponding to the outcomes of incompatible observables without assigning any probability at all to their conjunction.

\begin{definition}
An \textit{additive class} for a nonempty set $X$ is a nonempty subset of the power set $\mathcal{P}(X)$ such that
	\begin{enumerate}
    \item for all $A\in\Sigma$, $X\backslash A\in \Sigma$, and
    \item for all disjoint $A,B\in\Sigma$, $A\cup B\in\Sigma$.
    \end{enumerate}
\end{definition}
\noindent Notice that the only way that the definition of an additive class differs from that of an algebra is the restriction to disjoint sets in clause 2.  

\begin{definition}
A \textit{generalized probability space} is a triple $(X,\Sigma,\mu)$, where $\Sigma$ is an additive class for the nonempty set $X$ and $\mu:\Sigma\to[0,1]$ is such that
	\begin{enumerate}
    \item $\mu(X) = 1$, and
    \item for all disjoint $A,B\in \Sigma$, $\mu(A\cup B) = \mu(A) + \mu(B)$.
    \end{enumerate}
\end{definition}

Even though the Kolmogorov axioms 1 and 2 carry over exactly intact, generalized probability spaces are different from classical probability spaces in that the measure $\mu$ need not assign any probability at all to the union or intersection of events in $\Sigma$ that are not disjoint.  
\citet{Ma12} has found a generalized probability hidden variable representation for the Bell-EPR setup \citep{ClHoShHo69}.  
So it may seem as if generalized probability spaces escape the purview of Bell's theorem.  
Nevertheless, theorem \ref{thm:main} shows generalized probability spaces do not escape the KS theorem.

\begin{prop}
Every generalized probability space is a weak probability space.
\end{prop}

\begin{cor}
There exist simple quantum mechanical experiments with no weak hidden variable representation on a generalized probability space satisfying both WC and $\neg$FNC.\footnote{See \citet{Fe15} for an extensive discussion of generalized probability spaces and a slight variation on the main theorem applied to them.}
\end{cor}


\subsection{Extended Probability Spaces}
\label{sec:EPS}

Many authors have proposed extending the range of values that a probability measure can take in order to accommodate the predictions of quantum mechanics, especially those implicated in testing Bell-type inequalities.  
There are two obvious ways to do so: first, one can allow the probability measure to take on negative values, extending its range from $[0,1]$ to $\mathbb{R}$. 

\begin{definition}
A \textit{negative probability space} is a triple $(X,\Sigma,\mu)$, where $\Sigma$ is an algebra for the nonempty set $X$ and $\mu:\Sigma\to\mathbb{R}$ is such that
	\begin{enumerate}
    \item $\mu(X) = 1$, and 
    \item for all disjoint $A,B\in\Sigma$, $\mu(A\cup B) = \mu(A) + \mu(B)$.
    \end{enumerate}
\end{definition}

While the Kolmogorov axioms 1 and 2 carry over intact as before, negative probability spaces differ from classical probability spaces only in the range of the measure. 
This idea has a distinguished pedigree in quantum physics, going back at least to \citet{Di42} but also \citet{Fe48,Fe87}, especially in the context of the latter's path integral formulation of quantum mechanics.
In these approaches, negative probabilities are taken as convenient intermediates in calculations, hence are not clearly intended to have any interpretive significance in a hidden variable theory.
But other authors have explicitly developed negative probability hidden variable models for a variety of systems, often to escape the usual conclusions of the EPR \citeyearpar{EPR35} and \citet{Be64,Be71} results.\footnote{See, for example, \citet{Mu82,Mu86,Po88,Wo88,HoLeSe91,AgHoSc92,SuRo93,RoSu01,ScWaSc94,HaHwKo96}; and \citet{Kr05,Kr07,Kr08,Kr09}.}$^,$\footnote{
\citet{Kr05,Kr07,Kr08,Kr09} requires that the range of $\mu$ is contained in $[-1,\infty)$, calling the resulting spaces \textit{non-monotonic probability spaces}.
This restriction makes no difference for our remarks.}
(See \citet{Mu86} and \citet[\S3.5]{HoSe91} for reviews of some of these approaches.)

A variation on this idea is to use $p$-adic probability distribution for the hidden variable theory \citep{Kh95a,Kh95b,Kh09}. 
Given any rational number $x$, one can represent $x = p^a r/s$, where $p$ is prime and $r,s$ are integers not divisible by $p$.
Then the $p$-adic norm of $x$ is $|x|_p = p^{-a}$, which intuitively measures ``how divisible'' $x$ is by $p$.
The $p$-adic numbers are the Cauchy completion of the rational numbers with the $p$-adic norm instead of the usual Euclidean norm.
It turns out that negative rational numbers can respresent limiting relative frequencies on this theory, which in turn generate a $p$-adic probability measure, so that some probabilities can indeed be represented as negative in this sense \citep{Kh95a,Kh95b,Kh09}.

Finally, in work with both practical and foundational ambitions, \citet{Ha04,Ha08} and \citet{GMHa12} argue for the inclusion of negative probabilities in a formulation of the decoherent histories approach to quantum mechanics.\footnote{
In this context especially, \citet{HaYe13} suggest restricting attention to negative probability spaces, which they call \textit{quasiprobability distributions}, whose marginal distributions match those of some Kolmogorovian probability space.  
This restriction makes no difference for our remarks.}
They address the problem of interpreting negative probability values as follows: first, they point out that the probabilities dictated by a physical theory instruct (rational) agents on how to bet on outcomes of phenomena under that theory's purview.
Standard arguments from subjective Bayesian probability theory, the so-called Dutch book arguments, seem to demand that the Kolmogorovian axioms for classical probability theory must hold for the degrees of belief of any rational agent, which determine which bets the agent regards as fair.
However, they point out that these arguments only apply to bets which are \textit{settleable}, that is, bets about events, knowledge of which the agent in question will certainly know at some point. 
They then argue that only bets on sufficiently coarse-grained alternative histories will be settleable, where this coarse-graining guarantees that the alternatives' probabilities will lie in the unit interval.
However, because \citet{Ha08} and \citet{GMHa12} demand in their decoherent histories framework a single real fine-grained history over a set of preferred variables, it can be reformulated as a noncontextual hidden variable theory to which theorem \ref{thm:main} applies.  
We discuss this application at the end of this subsection.

In the second approach to extending the range of the probability measure, one goes even further, allowing it to take on complex values, as proposed by \citet{Iv78,Yo91,Yo94,Yo95,Yo96,Mi96,SrSu94}; and \cite{Sr95,Sr98,Sr06}.\footnote{\citet{Sr06} also considers \textit{quaternionic} probability spaces, which he sometimes calls \textit{extended} measure spaces.  The same results below apply to these spaces as well.}

\begin{definition}
A \textit{complex probability space} is a triple $(X,\Sigma,\mu)$, where
$\Sigma$ is an algebra for the nonempty set $X$ and $\mu:\Sigma\to\mathbb{C}$ is such that
	\begin{enumerate}
    \item $\mu(X) = 1$, and
    \item for all disjoint $A,B\in\Sigma$, $\mu(A\cup B) = \mu(A) + \mu(B)$.
    \end{enumerate}
\end{definition}

Just as with negative probability spaces, complex probability spaces differ from classical probability spaces only in the range of the measure.  
\citeauthor{Yo91} shows that complex probability spaces arise from a generalization of the \citet{Co46} axioms for probability that, among other virtues, explicitly allow one to avoid the standard conclusions of Bell's theorem against local, realistic theories \citep{Yo94,Yo95}.
\citeauthor{Mi96} also arrives at complex probability spaces through a consideration foundational issues, in particular of time-symmetric formulations of quantum mechanics involving weak measurement.
\citeauthor{Sr95}, meanwhile, is motivated by giving a measure-theoretic treatment to the complex values taken on by path integrals.

But application of theorem \ref{thm:main} shows that, despite their attractive features, the pathological features of negative and complex probability spaces will be quite a bit stronger.

\begin{prop}
Every negative and complex probability space is a weak probability space.
\end{prop}

\begin{cor}
There exist simple quantum mechanical experiments with no weak hidden variable representation on a negative or complex probability space satisfying both WC and $\neg$FNC.
\end{cor}

As mentioned above, this corollary implicates versions of the decoherent histories formalism that demand a unique fine-grained history \citep{Ha08,GMHa12} on a preferred set of values, like field configurations or particle positions.
Briefly, this framework postulates a set of mutually exclusive and exhaustive maximally fine-grained histories of these preferred values.
This set becomes the sample space for a negative probability distribution.\footnote{
\citet{Ha04} calls probability assignment lying outside the unit interval \textit{virtual probabilities}, and \cite{Ha08} and \citet{GMHa12} call them \textit{extended probabilities}, but the formal definition is the same.}
In particular, at each time there is a negative probability distribution for the preferred variables.
Because this distribution typically is intended to reproduce all the predictions that standard quantum mechanics would make, including the predictions for KS witnesses, the corollary applies.

Previously, \citet{BaGh99,BaGh00} had argued that the decoherent histories framework in general would run afoul of the KS theorem if it was committed to its assignments of properties in histories (or probabilities thereof) being ``objective'' elements of reality, i.e., to them being independent of the choice of set of mutually exclusive and exhaustive properties.
As \citet{Gr00a,Gr00b} pointed out, however, most versions of the decoherent (or consistent) histories approach make these value assignments only relative to such a choice of set; different value assignments using different sets cannot be combined.
This is called the \textit{single framework} (or \textit{family}) \textit{rule}.
In the version of decoherent histories proposed by \citet{Ha08} and \citet{GMHa12}, the arguments of \citet{BaGh99,BaGh00} do not strictly apply, for one set of variables determining mutually exclusive and exhaustive properties is privileged.
However, our results show that the KS theorem still has a bearing on this version.
We will discuss these implications for the decoherent histories approach further in \S\ref{sec:FNC}.

\subsection{Upper (Lower) Probability Spaces}
\label{sec:U(L)PS}

Suppes and collaborators \citep{SuZa91,dBSu00,dBSu10,HaSu10} have suggested that a slightly different type of hidden variable theory, using so-called upper probability spaces, does not run afoul of Bell-type theorems and thus may be a viable approach to the completion of quantum theory.  
This possibility seems to arise from the relaxation of one property of a classical probability measure: instead of being additive on disjoint measurable sets, an upper probability is only subadditive.
\begin{definition}
An \textit{upper (lower) probability space} is a triple $(X,\Sigma,\mu)$, where $\Sigma$ is an algebra for the nonempty set $X$ and $\mu:\Sigma\rightarrow[0,1]$ is such that
	\begin{enumerate}
	\item $\mu(X) = 1$, and
    \item for all disjoint $A,B\in\Sigma$, $\mu(A\cup B)\leq (\geq) \; \mu(A) + \mu(B)$.
	\end{enumerate}
\end{definition}

The only way that upper and lower probability spaces differ from classical probability spaces is in clause 2, requiring only subadditivity or superadditivity.  
An additional condition that upper probability measures are often demanded to satisfy is \textit{monotonicity}: for any $A,B \in \Sigma$, if $A \subseteq B$ then $\mu(A) \leq \mu(B)$.
When an upper probability measure satisfies this condition then it generates a lower probability measure in a canonical way \citep{SuZa91}.
The upper and lower probabilities assigned to events together constitute an interval-valued probability assignment.
Such spaces have been used in other contexts in psychology and behavioral science for modeling agents with imprecise beliefs \citep{Br15}.
\citet{SuZa91} have suggested that they may also be useful for modeling quantum systems, but only if monotonicity is dropped.  
In particular, Suppes and collaborators have constructed concrete upper probability hidden variable theory for both the EPR \citep{HaSu10} and GHZ \citep{dBSu00,dBSu10} experiments.  
This seems to suggest that hidden variable theories modeled on upper probability spaces avoid Bell-type no-go results.  
But, again, it is easy to see that theorem \ref{thm:main} implies that upper probability spaces cannot avoid the KS theorem.

\begin{prop}
Every upper (lower) probability space is a weak probability space.
\end{prop}

\begin{cor}
There exist simple quantum mechanical experiments with no weak hidden variable representation on an upper (lower) probability space satisfying both WC and $\neg$FNC.
\end{cor}

This shows that any upper probability hidden variable theory for quantum mechanics must violate either WC or $\neg$FNC when it represents a KS witness.  
Although the examples considered by Suppes and collaborators for the EPR \citep{HaSu10} and GHZ \citep{dBSu00,dBSu10} experiments themselves do not contain enough possible measurements to be KS witnesses, the measurements they do concern are subsets of KS witnesses.  
Thus, since these examples satisfy WC, one can avoid contradiction with the KS theorem by finding a finite null cover for the space.  
To do this, in appendix \ref{sec:appendix-upper} we explicitly construct a finite null cover consistent with each of their examples.  One might accept that upper and lower probability spaces will lead to certain kinds of pathologies because, as \citet{SuZa91} already appear to accept, these models must be nonmonotonic.  But our result implies that to use upper and lower probability spaces for quantum mechanics, one must additionally accept the pathology of a finite null cover, which may be less palatable.

\subsection{Quantum Measure Theory}
\label{sec:QMT}
\citet{So94,So97} proposes an alternative probability theory for use in quantum gravity and quantum cosmology.  
He claims that quantum mechanics is a special case of an alternative probability theory in which the additivity requirement is weakened to allow for pairwise, but no higher-order, interference.  
These he calls \textit{quantum measure spaces}. 
(See \citet{Gu10a,Gu10b} for a mathematical treatment, which we follow.)

\begin{definition}
A \textit{quantum measure space}\footnote{\citet{Gu10b} also includes an extra condition called \textit{regularity}.  For our purposes it suffices to consider the simpler and more general spaces defined here, of which Gudder's are a special case.} is a triple $(X,\Sigma,\mu)$, where $\Sigma$ is an algebra for the nonempty set $X$ and $\mu: \Sigma\to[0,1]$ is such that
	\begin{enumerate}
	\item $\mu(X) = 1$, and
    \item for all pairwise disjoint $A,B,C\in\Sigma$,
    \[\mu(A\cup B\cup C) = \mu(A\cup B) + \mu(B\cup C) + \mu(A\cup C) - \mu(A) - \mu(B) - \mu(C).\]
	\end{enumerate}
\end{definition}

The only way that quantum measure spaces differ from classical probability spaces is in clause 2.  
While not requiring the usual additivity axiom, quantum measure spaces instead require a weaker condition amounting to additivity on triples of events.\footnote{
One can also generalize to even weaker additivity constraints, to which our results apply equally.}  
The failure of additivity on pairs of events allows for the possibility of pairwise interference.  
So, Sorkin claims, one can accurately represent quantum experiments like the triple slit setup, where quantum mechanics predicts pairwise interference but strict additivity on triples.

Since Sorkin's suggestion, there has been quite a bit of work on quantum measure theory, including using it as a model for a noncontextual hidden variable theory.
For example, \citet[p. 519]{CDHMRS07} explicitly suggest that
``it is possible to attribute the correlations in the EPRB setup to non-contextual hidden variables, so long as they are \textit{quantal} hidden variables,'' that is, hidden variables described using quantum measure theory rather than classical probability theory.
Moreover, previous work on the KS theorem in the context of quantum measure theory revealed specific cases where a finite null cover arose for certain quantum measures designed to recover the empirical predictions of quantum mechanics \citep{DoGT08,SuWa10}. 
But theorem \ref{thm:main} allows us to conclude something quite a bit stronger---that finite null covers are in a certain sense \textit{generic} in quantum measure theories, in that they \textit{must} appear in any quantum measure theoretic representation of a KS witness.  First, we state our result with the same logical structure as the previous propositions.

\begin{prop}
Every quantum measure space is a weak probability space.
\end{prop}

\begin{cor}
There exist simple quantum mechanical experiments with no weak hidden variable representation on a quantum measure space satisfying both WC and $\neg$FNC.
\end{cor}

\citet{SuWa10} investigated a related question, whether an event $A \in \Sigma$ of a quantum measure space admits of a (nontrivial) \textit{quantum cover}.  A quantum cover for $A\in\Sigma$ is a collection $\{A_i\} \subseteq \Sigma$ such that $A \subseteq \bigcup_i A_i$ and for any quantum measure $\mu$, $\mu(A_i)=0$ for all $i$ implies $\mu(A)=0$.
They were able to show that every non-atomic event of the algebra $\Sigma$, in particular the whole space $X$, possesses a quantum cover if $X$ is finite in cardinality.

Our results, by contrast, can be restated as follows:
\begin{prop}
If a KS witness has a weak hidden variable representation on a quantum measure space $(X,\Sigma,\mu)$, then for each $A \in \Sigma$ not all covers of $A$ can be quantum covers.  In particular, not all covers of $X$ can be quantum covers.
\end{prop}
Thus, even if $X$ has \textit{some} quantum cover, our result shows that $X$ must also have some cover that fails to be a quantum cover.  \citet{DoGT08} and \citet{SuWa10} have advocated for an interpretation of quantal measure-zero events as being ``precluded'' from occurring and argued that the problem with classical probability models of noncontextual hidden variable theories is that they allow for events covering the sample space to be precluded in KS witnesses.
However, the above proposition shows that moving to quantum measure spaces is no help: they too must have a set of precluded events that cover the sample space.

\section{Against Finite Null Covers}
\label{sec:FNC}

We believe that the condition WC is unobjectionable, for the use of classical probability theory in experiments without noncommuting observables is an anodyne feature of scientific practice.  
To save the possibility of noncontextual hidden variable theories, one might therefore be tempted to reject $\neg$FNC. 
However, it is not obvious what a physical interpretation of a finite null cover could be.
We now show a further feature of finite null covers that make such a rejection at the very least unattractive.  
Specifically, we show that a finite null cover involves a kind of pathology or irrationality for those who use it to set their degrees of belief.

A physical theory involving probabilities usually aspires to instruct its users on how to set their expectations to outcomes of experiments within the theory's purview, if the theory is to be regarded as correct.
If such a theory dictates that the probability of a certain event is $p$, a bet that pays out $s$ (in dollars, say,) would have a fair price of $ps$: taking such a bet would on average not yield any net gain or loss if the agent's degrees of belief are a classical probability measure.
Thus, one way to describe betting irrationalities in probability theories is through the concept of a \textit{Dutch book}, a finite sequence of bets that the agent regards as fair but which is guaranteed to yield a loss for the agent.

For our purposes, a \textit{Dutch book} in a weak probability space $(X,\Sigma,\mu)$ is a collection of events $B_1,...,B_m\in\Sigma$ along with numbers $s_1,...,s_m\in\mathbb{R}$ such that for any $x\in X$,
\[\sum_{i=1}^m s_i(\chi_{B_i}(x)-\mu(B_i))<0,\]
where $\chi_{B_i}(x)$ is the characteristic function of $B_i$, i.e.,
\begin{equation*}
\chi_{B_i}(x) = \begin{cases}
1, \text{ if } x \in B_i, \\
0, \text{ if } x \notin B_i.
\end{cases}
\end{equation*}
The events $B_i$ are those on which an agent is betting, with the numbers $s_i$ representing the payout for a winning bet on $B_i$ (and $s_i \mu(B_i)$ representing the bet's fair buy-in).
If a weak probability space contains a Dutch book, then an agent who sets their beliefs according to that probability space will accept as fair a series of bets on which she is guaranteed to lose money.  
Hence, a Dutch book illustrates a kind of irrationality or pathology in the probability space in question.

\begin{prop}
If a weak probability space $(X,\Sigma,\mu)$ violates $\neg$FNC, then it contains a Dutch book.\footnote{See also \citet{Fe15} for a discussion of Dutch Books in the context of generalized probability spaces.}
\end{prop}

\begin{proof}
Let $B_1,...,B_m$ be a finite null cover for $X$.  
Let $s_i < 0$ for all $i \in \{1,\ldots,m\}$.  
Then for any $x\in X$, we know that $x\in B_i$ for some $i \in \{1,\ldots,m\}$ so
\[\sum_{i=1}^m s_i(\chi_{B_i}(x)-\mu(B_i)) = \sum_{i=1}^m s_i\chi_{B_i}(x) < 0.\]
\end{proof}

This shows that whenever a weak probability space violates $\neg$FNC, anyone who sets their beliefs according to that probability space will be irrational in the sense that she will accept as fair a series of bets on which she is guaranteed to lose money.

Is there any way around this deficiency?
In the context of negative probability spaces, \citet{Ha04} and \citet{GMHa12} consider the problem of whether a Dutch book can be made against an agent who assigns negative probabilities to some histories.
They point out that certain events will not be \textit{settleable} by an agent, i.e., their outcome will not be recorded in any records in the agent's history.
Under the appropriate decoherence conditions, negative probability events will not be settleable.
Although the events in a finite null cover are not assigned negative probability, one might question whether they are all settleable; if they are not, then an agent setting her degrees of belief according to them would not be susceptible to a Dutch book.
However, each individual event does seem to be setttleable: an examination of the proof of theorem \ref{thm:main} shows that the events in question can be constructed by finitely many logical operations from the measurements determined by \textit{compatible} projection operators in a simple quantum mechanical experiment.
That is, for any individual probability zero event in a finite null cover, one could do an experiment to check whether this event occurs.

It might be protested that despite this, several of the events together are not settleable since the projection used to generate one probability zero event in the finite null cover may fail to commute with the projection used to generate another probability zero event in the finite null cover.
Thus, one might object that one cannot always jointly settle some pairs of events in a finite null cover.
However, the union of the entire collection of probability zero events in the finite null cover must be settleable, for settling this event---namely, checking whether any of the individual probability zero events occurs---is equivalent to settling the question of whether \textit{anything} happens.
A persistent objector must then demand that a highly nonclassical logic of betting: that placing a finite sequence of bets does not commit one to a bet on their disjunction.
This is impossible with a classical disjunction, which is true just in case on of its disjuncts is true.

So we have not yet found satisfactory physical interpretation of finite null covers, and we conclude that a finite null cover appears to represent a kind of irrationality in a nonclassical probability theory.  
We stress that the problem with finite null covers does not only arise in a subjective belief interpretation of probability theory---similar considerations weigh against, for example, relative frequency interpretations of finite null covers as well because one would have to make sense of a finite number of relative frequency zero events who collectively have relative frequency one.  
Thus, we conclude that a finite null cover indicates a pathology in any probability theory, regardless of the interpretation of probability we choose.
Noncontextual, non-Kolmogorovian hidden variable theories therefore do not clearly offer any advantages over contextual theories.

\section{Conclusion and Future Directions}

We have considered a variety of generalizations of probability theory that researchers have proposed in order to construct hidden variable models for quantum mechanics.  
These alternative probability theories---including generalized probability spaces, extended probability spaces, upper/lower probability spaces, and quantum measure spaces---all have the advantage of avoiding Bell's theorem.  
We showed that even though these generalizations of probability theory avoid Bell's theorem, they nevertheless fall prey to another no-go theorem due to Kochen and Specker.  
The KS theorem implies that if we apply any of these alternative probability theories in quantum mechanics, then it must violate one of two weak constraints---it must either fail to be weakly classical, or else contain a finite null cover.  
We submit that the constraint of weak classicality is an unobjectionable precondition to making sense of any probabilistic experiment in quantum physics, and weak classicality has been implicitly accepted by all the authors considered.  
On the other hand, a finite null cover is a kind of pathology in any alternative probability theory.  
So the result of this paper can be understood as a no-go theorem for the use of these alternative probability theories for quantum mechanics.

There is still further research to be done in this area.  
First, there are other alternative probability theories that have been proposed for quantum mechanics that are not yet known to fall under the purview of our results---for example, the \textit{prism models} of \citet{Fi82c}, the \textit{deterministic nonmeasurable models} of \citet{Pi83}, the \textit{generic models} of \citet{We06a,We06b}, and perhaps others.  
Further work is required to see whether the results of this paper can someone be extended to rule out even more alternative probability theories that have been proposed for quantum mechanics.  
Second, the results of this paper show a logical relation between the KS theorem and Bell's theorem: the KS theorem is, in a certain sense, stronger than Bell's theorem, for the former rules out hidden variable models that the latter does not, but not vice versa.  
However, one finds a plethora of other no-go theorems in the literature, including the recent PBR theorem \citep{PuBaRu12}.  
It would be interesting to investigate the logical relations between the KS theorem, Bell's theorem, and other existing no-go results.  
Is there one no-go theorem whose witnesses encompass all the others'?

\appendix
\section{Null Covers for Some Upper Probability Space Hidden Variable Models}
\label{sec:appendix-upper}
\subsection{EPR Setup}
In the Bell-EPR setup \citep{ClHoShHo69}, Alice and Bob have two possible ``yes/no" measurements they can make on the Bell state.  
Let $A,A'$ and $B,B'$, respectively, denote the events that these measurements return the value ``yes".  
The projection operators corresponding with each element of $\{A,A'\}$ commutes with that of each element of $\{B,B'\}$, but that of $A$ ($B$) does not commute with that of $A'$ ($B'$). 
Despite this non-commutativity, \citet{HaSu10} propose to assign an \textit{upper probability} to the algebra of events generated from these four that also reproduces the predictions of quantum mechanics.  
That is, they consider an upper probability space $(X,\Sigma,\mu)$ and assign the following upper probabilities $\mu(C)$ to the atomic events $C$ of $X$, where for any $S \in \Sigma$, $S^c = X \backslash S$:
\begin{center}
\begin{tabular}{r|cccc}
$\cap$ 				& $A \cap A'$ 				& $A^c \cap A'$ 				& $A \cap A'^c$ 				& $A^c \cap A'^c$ \\ \hline 
$B \cap B'$ 		& $0$ 							& $\frac{1}{16}$ 					& $\frac{1}{8}$						& $\frac{1}{8} + \frac{\sqrt{3}}{8}$ \\
$B^c \cap B'$ 	& $\frac{1}{8}$					& $\frac{1}{8} - \frac{\sqrt{3}}{8}$& $0$								& $\frac{1}{16}$ \\
$B \cap B'^c$	& $\frac{1}{16}$				& $0$								& $\frac{1}{8} - \frac{\sqrt{3}}{8}$& $\frac{1}{8}$ \\
$B^c \cap B'^c$ & $\frac{1}{8} + \frac{\sqrt{3}}{8}$& $\frac{1}{8}$				& $\frac{1}{16}$					& $0$
\end{tabular}
\end{center}
They also assign the following (partial) joint probabilities:
\begin{center}
\begin{tabular}{r|cccc}
$\cap$ 	& $A$ 								& $A^c$ 							& $A'$ 								& $A'^c$ \\ \hline
$B$		& $\frac{1}{4} - \frac{\sqrt{3}}{8}$& $\frac{1}{4} + \frac{\sqrt{3}}{8}$& $0$								& $\frac{1}{2}$\\
$B^c$	& $\frac{1}{4} + \frac{\sqrt{3}}{8}$& $\frac{1}{4} - \frac{\sqrt{3}}{8}$& $\frac{1}{2}$						& $0$ \\
$B'$	& $\frac{1}{8}$						& $\frac{3}{8}$						& $\frac{1}{4} - \frac{\sqrt{3}}{8}$& $\frac{1}{4} + \frac{\sqrt{3}}{8}$\\
$B'^c$	& $\frac{3}{8}$						& $\frac{1}{8}$						& $\frac{1}{4} + \frac{\sqrt{3}}{8}$& $\frac{1}{4} - \frac{\sqrt{3}}{8}$
\end{tabular}
\end{center}
In addition, they set 
\begin{equation*}
\mu(A) = \mu(A^c) = \mu(B) = \mu(B^c) = \mu(A') = \mu(A'^c) = \mu(B') = \mu(B'^c) = 1/2.
\end{equation*}
They verify that these assignments are consistent and reproduce the predictions of quantum mechanics for a certain Bell-EPR experiment, even though they do not uniquely specify an upper probability space, which requires $2^{16} - 2 = 65,534$ assignments to make. 
Note as well that the symmetry of the EPR state requires that these assignments are invariant under the atomic complementation map.

Theorem \ref{thm:main} implies that the upper probability space hidden variable theory described above must contain a finite null cover if it is not inconsistent, since it satisfies WC.  The following additional assignments witness that possibility:\footnote{Such an additional assignment to realize a finite null cover does not appear to be unique in general, although it may be so in specific cases.}
\begin{gather*}
\mu(A^c \cap ((A'^c \cap B) \cup (A' \cap B^c \cap B') )) = \mu(A \cap ((A' \cap B^c) \cup (A'^c \cap B \cap B'^c) )) = 0, \\
\mu(A^c \cap ((A' \cap B^c) \cup (A'^c \cap B \cap B'^c))) = \mu(A \cap ((A'^c \cap B) \cup (A' \cap B^c \cap B'))) = 0.
\end{gather*}
Note that these assignments are also preserved under the atomic complementation map.  
To show this assignment is consistent with the previous assignments, we must check that $\mu$ is still subadditive on disjoint sets.  
This is automatic for the decomposition of these newly assigned null sets.  
The only superset of any of these null sets that is assigned a measure so far is the whole space.  
So, for example:
\begin{align*}
1 &= \mu(X) \\
&\leq \mu(A^c \cap ((A'^c \cap B) \cup (A' \cap B^c \cap B') )) + \mu((A^c \cap ((A'^c \cap B) \cup (A' \cap B^c \cap B') ))^c) \\
&= \mu((A^c \cap ((A'^c \cap B) \cup (A' \cap B^c \cap B') ))^c) \\
&\leq \sum_{\text{atomic } C \in X} \mu(C) \\
& \qquad - \mu(A^c \cap A'^c \cap B \cap B') - \mu(A^c \cap A'^c \cap B \cap B'^c) - \mu(A^c \cap A' \cap B^c \cap B') \\
&= 1 + \frac{3}{8} - \left(\frac{1}{8} + \frac{\sqrt{3}}{8}\right) - \frac{1}{8} - \left(\frac{1}{8} - \frac{\sqrt{3}}{8}\right) = 1,
\end{align*}
where in the last line we have used the calculation at the bottom of p. 97 of \cite{HaSu10}.  
The remaining verifications are similar.  
Since subadditivity holds for this decomposition into atoms, it holds for decomposition into larger subsets as well.

\subsection{GHZ Setup}
For the GHZ setup \citep{GrHoZe89,GrHoShZe90,Me90}, there are three ``yes/no" observables, none of which pairwise commute.
Let $A,B,C$ denote the events that these observables return the value ``yes."
In order to reproduce the predictions of quantum mechanics, \citet{dBSu10} make the following assignments to the atoms of an upper probability space hidden variable model $(X,\Sigma,\mu)$ for this setup, where again $S^c = X \backslash S$ for any $S \in \Sigma$:
\begin{gather*}
\mu(A \cap B \cap C) = \mu(A^c \cap B^c \cap C^c) = 1, \\
\mu(A^c \cap B \cap C) = \mu(A \cap B^c \cap C) = \mu(A \cap B \cap C^c) = 0, \\
\mu(A \cap B^c \cap C^c) = \mu(A^c \cap B \cap C^c) = \mu(A^c \cap B^c \cap C) = 0.
\end{gather*}
They also make the following (partial) joint assignments:
\begin{gather*}
\mu(A) = \mu(B) = \mu(C) = 1, \\
\mu(A^c) = \mu(B^c) = \mu(C^c) = 0.
\end{gather*}
Since their assignments are designed to satisfy WC and reproduce the quantum mechanical expectations, on pain of contradiction they must be compatible with the existence of a finite null cover by theorem \ref{thm:main}.  
Their assignments do not completely determine the values the upper probability measure takes on all $2^8-2=254$ nontrivial measurable sets, but by fixing the following assignment the finite null cover is achieved:
\begin{equation*}
\mu((A \cap B \cap C) \cup (A^c \cap B^c \cap C) \cup (A \cap B^c \cap C^c) \cup (A^c \cap B \cap C^c) ) = 0.
\end{equation*}
To show this assignment is consistent with the previous assignments, we must check that the measure is still subadditive on disjoint sets.  
This is automatic for its decomposition into disjoint sets because it is null. 
The only superset of this null set assigned a measure so far is the whole space:
\begin{align*}
1 &= \mu(X) \\ 
&\leq \mu((A \cap B \cap C) \cup (A^c \cap B^c \cap C) \cup (A \cap B^c \cap C^c) \cup (A^c \cap B \cap C^c)) \\
& \quad + \mu(((A \cap B \cap C) \cup (A^c \cap B^c \cap C) \cup (A \cap B^c \cap C^c) \cup (A^c \cap B \cap C^c))^c) \\
&= \mu(((A \cap B \cap C) \cup (A^c \cap B^c \cap C) \cup (A \cap B^c \cap C^c) \cup (A^c \cap B \cap C^c))^c).
\end{align*}
Hence
\begin{align*}
1 &= \mu(((A \cap B \cap C) \cup (A^c \cap B^c \cap C) \cup (A \cap B^c \cap C^c) \cup (A^c \cap B \cap C^c))^c) \\
&= \mu((A^c \cap B^c \cap C^c) \cup (A \cap B \cap C^c) \cup (A^c \cap B \cap C) \cup (A \cap B^c \cap C)) \\
&\leq \mu(A^c \cap B^c \cap C^c) + \mu(A \cap B \cap C^c) + \mu(A^c \cap B \cap C) + \mu(A \cap B^c \cap C) \\
&= 1 + 0 + 0+ 0 =1.
\end{align*}

\bibliographystyle{apalike}
\setcitestyle{authoryear}
\bibliography{quantumprobability}

\begin{thebibliography}{}

\bibitem[Agarwal et~al., 1992]{AgHoSc92}
Agarwal, G., Home, D., and Schleich, W. (1992).
\newblock {Einstein}-{Podolsky}-{Rosen} correlation---parallelism between the
  {Wigner} function and the local hidden variable approaches.
\newblock {\em Physics Letters A}, 170:359--362.

\bibitem[Aspect, 1983]{As83}
Aspect, A. (1983).
\newblock {\em Trois tests exp\'{e}rimentaux des in\'{e}galit\'{e}s de {Bell}
  par corr\'{e}lation de polarisation de photons}.
\newblock PhD thesis, Universite de Paris-Sud, Cenre d'Orsay, Orsay.

\bibitem[Aspect et~al., 1982a]{AsDaRo82}
Aspect, A., Dalibard, J., and Roger, G. (1982a).
\newblock Experimental test of {Bell's} inequalities using time-varying
  analyzers.
\newblock {\em Physical Review Letters}, 49:1804--1807.

\bibitem[Aspect et~al., 1982b]{AsGrRo82}
Aspect, A., Grangier, P., and Roger, G. (1982b).
\newblock Experimental realization of {Einstein}-{Podolsky}-{Rosen}-{Bohm}
  gedankenexperiment: a new violation of {Bell's} inequalities.
\newblock {\em Physical Review Letters}, 49:91--94.

\bibitem[Bassi and Ghirardi, 1999]{BaGh99}
Bassi, A. and Ghirardi, G. (1999).
\newblock Can the decoherent histories description of reality be considered
  satisfactory?
\newblock {\em Physics Letters A}, 257:247--263.

\bibitem[Bassi and Ghirardi, 2000]{BaGh00}
Bassi, A. and Ghirardi, G. (2000).
\newblock Decoherent histories and realism.
\newblock {\em Journal of Statistical Physics}, 98:457--494.

\bibitem[Bell, 1964]{Be64}
Bell, J. (1964).
\newblock On the {Einstein} {Podolsky} {Rosen} paradox.
\newblock {\em Physics}, 1(3):195--200.

\bibitem[Bell, 1966]{Be66}
Bell, J. (1966).
\newblock On the problem of hidden variables in quantum mechanics.
\newblock {\em Reviews of Modern Physics}, 38(3):447--451.

\bibitem[Bell, 1971]{Be71}
Bell, J. (1971).
\newblock Introduction to the hidden-variable question.
\newblock In d'Espagnat, B., editor, {\em Foundations of Quantum Mechanics
  (Proceedings of the International School of Physics 'Enrico Fermi', course
  IL)}, pages 171--181, New York. Academic Press.

\bibitem[Bradley, 2015]{Br15}
Bradley, S. (2015).
\newblock Imprecise probabilities.
\newblock In Zalta, E., editor, {\em The Stanford Encyclopedia of Philosophy}.

\bibitem[Clauser and Horne, 1974]{ClHo74}
Clauser, J. and Horne, M. (1974).
\newblock Experimental consequences of objective local theories.
\newblock {\em Physical Review D}, 10:526--535.

\bibitem[Clauser et~al., 1969]{ClHoShHo69}
Clauser, J., Horne, M., Shimony, A., and Holt, R. (1969).
\newblock Proposed experiment to test local hidden-variable theories.
\newblock {\em Physical Review Letters}, 23:880--884.

\bibitem[Cox, 1946]{Co46}
Cox, R. (1946).
\newblock Probability, frequency and reasonable expectation.
\newblock {\em American Journal of Physics}, 14:1--10.

\bibitem[Craig et~al., 2007]{CDHMRS07}
Craig, D., Dowker, F., Henson, J., Major, S., Rideout, D., and Sorkin, R.
  (2007).
\newblock A {Bell} inequality analog in quantum measure theory.
\newblock {\em Journal of Physics A}, 40:501--523.

\bibitem[de~Barros and Suppes, 2000]{dBSu00}
de~Barros, J.~A. and Suppes, P. (2000).
\newblock Some conceptual issues involving probability in quantum mechanics.
\newblock arXiv:quant-ph/0001017.

\bibitem[de~Barros and Suppes, 2010]{dBSu10}
de~Barros, J.~A. and Suppes, P. (2010).
\newblock Probabilistic inequalities and upper probabilities in quantum
  mechanical entanglement.
\newblock {\em Manuscrito --- Revista Internacional de Filosofia},
  33(1):55--71.

\bibitem[Dirac, 1942]{Di42}
Dirac, P. (1942).
\newblock {Bakerian} {Lecture}: The physical interpretation of quantum
  mechanics.
\newblock {\em Proceedings of the Royal Society of London A}, 180:1--40.

\bibitem[Dowker and Ghazi-Tabatabai, 2008]{DoGT08}
Dowker, F. and Ghazi-Tabatabai, Y. (2008).
\newblock The {Kochen}-{Specker} theorem revisited in quantum measure theory.
\newblock {\em Journal of Physics A}, 41:105301.

\bibitem[Einstein et~al., 1935]{EPR35}
Einstein, A., Podolsky, B., and Rosen, N. (1935).
\newblock Can quantum-mechanical description of physical reality be considered
  complete?
\newblock {\em Physical Review}, 47:777--780.

\bibitem[Feintzeig, 2015]{Fe15}
Feintzeig, B. (2015).
\newblock Hidden variables and incompatible observables in quantum mechanics.
\newblock {\em British Journal for the Philosophy of Science}, 66(4):905--927.

\bibitem[Feynman, 1948]{Fe48}
Feynman, R. (1948).
\newblock Space-time approach to non-relativistic quantum mechanics.
\newblock {\em Reviews of Modern Physics}, 20(2):367--387.

\bibitem[Feynman, 1987]{Fe87}
Feynman, R. (1987).
\newblock Negative probability.
\newblock In Hiley, B. and Peat, F.~D., editors, {\em Quantum Implications:
  Essays in Honour of David Bohm}, pages 235--248. Routledge, New York.

\bibitem[Fine, 1982a]{Fi82c}
Fine, A. (1982a).
\newblock Antinomies of entanglement: The puzzling case of the tangled
  statistics.
\newblock {\em The Journal of Philosophy}, 79(12):733--747.

\bibitem[Fine, 1982b]{Fi82a}
Fine, A. (1982b).
\newblock Hidden variables, joint probability, and the {Bell} inequalities.
\newblock {\em Physical Review Letters}, 48(5):291--294.

\bibitem[Fine, 1982c]{Fi82b}
Fine, A. (1982c).
\newblock Joint distributions, quantum correlations, and commuting observables.
\newblock {\em Journal of Mathematical Physics}, 23(7):1306--1310.

\bibitem[Gell-Mann and Hartle, 2012]{GMHa12}
Gell-Mann, M. and Hartle, J. (2012).
\newblock Decoherent histories quantum mechanics with one real fine-grained
  history.
\newblock {\em Physical Review A}, 85:062120.

\bibitem[Giustina et~al., 2013]{Gietal13}
Giustina, M., Mech, A., Ramelow, S., Wittman, B., Kofler, J., Beyer, J., Lita,
  A., Calkins, B., Gerrits, T., Nam, S., Ursin, R., and Zeilinger, A. (2013).
\newblock {Bell} violation using entangled photons without the fair-sampling
  assumption.
\newblock {\em Nature}, 497:227--230.

\bibitem[Giustina et~al., 2015]{Gietal15}
Giustina, M., Versteegh, M., Wengerowsky, S., Handsteiner, J., Hochrainer, A.,
  Phelan, K., Steinlechner, F., Kofler, J., Larsson, J., Abell\'{a}n, C.,
  Amaya, W., Pruneri, V., Mitchell, M., Beyer, J., Gerrits, T., Lita, A.,
  Shalm, L., Nam, S., Scheidl, T., Ursin, R., Wittmann, B., and Zeilinger, A.
  (2015).
\newblock Significant-loophole-free test of {Bell's} theorem with entangled
  photons.
\newblock {\em Physical Review Letters}, 115:250401.

\bibitem[Greenberger et~al., 1990]{GrHoShZe90}
Greenberger, D., Horne, M., Shimony, A., and Zeilinger, A. (1990).
\newblock {Bell's} theorem without inequalities.
\newblock {\em American Journal of Physics}, 58:1131--1143.

\bibitem[Greenberger et~al., 1989]{GrHoZe89}
Greenberger, D., Horne, M., and Zeilinger, A. (1989).
\newblock Going beyond {Bell's} theorem.
\newblock In Kafatos, M., editor, {\em Bell's Theorem, Quantum Theory and
  Conceptions of the Universe}, pages 69--72. Kluwer, Dordrecht.

\bibitem[Griffiths, 2000a]{Gr00a}
Griffiths, R. (2000a).
\newblock Consistent histories, quantum truth functionals, and hidden
  variables.
\newblock {\em Physics Letters A}, 265:12--19.

\bibitem[Griffiths, 2000b]{Gr00b}
Griffiths, R. (2000b).
\newblock Consistent quantum realism: A reply to {Bassi} and {Ghirardi}.
\newblock {\em Journal of Statistical Physics}, 99:1409--1425.

\bibitem[Gudder, 1988]{Gu88}
Gudder, S. (1988).
\newblock {\em Quantum Probability}.
\newblock Academic Press, San Diego, CA.

\bibitem[Gudder, 2010a]{Gu10a}
Gudder, S. (2010a).
\newblock Finite quantum measure spaces.
\newblock {\em The American Mathematical Monthly}, 117(6):512--527.

\bibitem[Gudder, 2010b]{Gu10b}
Gudder, S. (2010b).
\newblock Quantum measure theory.
\newblock {\em Mathematica Slovaca}, 60(5):681--700.

\bibitem[Halliwell and Yearsley, 2013]{HaYe13}
Halliwell, J. and Yearsley, J. (2013).
\newblock Negative probabilities, {Fine's} theorem, and linear positivity.
\newblock {\em Physical Review A}, 87:022114.

\bibitem[Han et~al., 1996]{HaHwKo96}
Han, Y.~D., Hwang, W., and Koh, I. (1996).
\newblock Explicit solutions for negative probability measures for all
  entangled states.
\newblock {\em Physics Letters A}, 221:283--286.

\bibitem[Hartle, 2004]{Ha04}
Hartle, J. (2004).
\newblock Linear positivity and virtual probability.
\newblock {\em Physical Review A}, 70:022104.

\bibitem[Hartle, 2008]{Ha08}
Hartle, J. (2008).
\newblock Quantum mechanics with extended probabilities.
\newblock {\em Physical Review A}, 78:012108.

\bibitem[Hartmann and Suppes, 2010]{HaSu10}
Hartmann, S. and Suppes, P. (2010).
\newblock Entanglement, upper probabilities and decoherence in quantum
  mechanics.
\newblock In Su\'{a}rez, M., Dorato, M., and R\'{e}dei, M., editors, {\em EPSA
  Philosophical Issues in the Sciences: Launch of the European Philosophy of
  Science Association}, pages 93--103. Springer.

\bibitem[Home et~al., 1991]{HoLeSe91}
Home, D., Lepore, V., and Selleri, F. (1991).
\newblock Local realistic models and non-physical probabilities.
\newblock {\em Physics Letters A}, 158:357--360.

\bibitem[Home and Selleri, 1991]{HoSe91}
Home, D. and Selleri, F. (1991).
\newblock {Bell's} theorem and the {EPR} paradox.
\newblock {\em Revista del Nuovo Cimento}, 14(9):1--95.

\bibitem[Ivanovi\'{c}, 1978]{Iv78}
Ivanovi\'{c}, I. (1978).
\newblock On complex {Bell's} inequality.
\newblock {\em Lettere al Nuovo Cimento}, 22(1):14--16.

\bibitem[Khrennikov, 1995a]{Kh95a}
Khrennikov, A. (1995a).
\newblock p-adic probability distributions of hidden variables.
\newblock {\em Physica A}, 215(4):577--587.

\bibitem[Khrennikov, 1995b]{Kh95b}
Khrennikov, A. (1995b).
\newblock p-adic probability interpretation of {Bell's} inequality.
\newblock {\em Physics Letters A}, 200:219--223.

\bibitem[Khrennikov, 2009]{Kh09}
Khrennikov, A. (2009).
\newblock {\em Interpretations of Probability}.
\newblock de Gruyter, 2nd edition.

\bibitem[Kochen and Specker, 1967]{KoSp67}
Kochen, S. and Specker, E. (1967).
\newblock The problem of hidden variables in quantum mechanics.
\newblock {\em Journal of Mathematics and Mechanics}, 17:59--87.

\bibitem[Krantz et~al., 1971]{KrLuSuTv71}
Krantz, D., Luce, D., Suppes, P., and Tversky, A. (1971).
\newblock {\em Foundations of Measurement}, volume~I.
\newblock Dover, Mineola, NY.

\bibitem[Kronz, 2005]{Kr05}
Kronz, F. (2005).
\newblock A nonmonotonic theory of probability for spin-1/2 systems.
\newblock {\em International Journal of Theoretical Physics},
  44(11):1963--1976.

\bibitem[Kronz, 2007]{Kr07}
Kronz, F. (2007).
\newblock Non-monotonic probability theory and photon polarization.
\newblock {\em Journal of Philosophical Logic}, 36:446--472.

\bibitem[Kronz, 2008]{Kr08}
Kronz, F. (2008).
\newblock Non-monotonic probability theory for n-state quantum systems.
\newblock {\em Studies in the History and Philosophy of Modern Physics},
  39:259--272.

\bibitem[Kronz, 2009]{Kr09}
Kronz, F. (2009).
\newblock Actual and virtual events in the quantum domain.
\newblock {\em Ontology Studies}, 9:209--220.

\bibitem[Malament, 2012]{Ma12}
Malament, D. (2012).
\newblock Notes on {Bell's} theorem.
\newblock
  \url{http://www.socsci.uci.edu/~dmalamen/courses/prob-determ/PDnotesBell.pdf}.

\bibitem[Mermin, 1986]{Me86}
Mermin, N. (1986).
\newblock Generalizations of {Bell's} theorem to higher spins and higher
  correlations.
\newblock In Roth, L. and Inomato, A., editors, {\em Fundamental Questions in
  Quantum Mechanics}, pages 7--20. Gordon and Breach, New York.

\bibitem[Mermin, 1990]{Me90}
Mermin, N. (1990).
\newblock Quantum mysteries revisited.
\newblock {\em American Journal of Physics}, 58:731--733.

\bibitem[Miller, 1996]{Mi96}
Miller, D. (1996).
\newblock Realism and time symmetry in quantum mechanics.
\newblock {\em Physics Letters A}, 222:31--36.

\bibitem[M\"{u}ckenheim, 1982]{Mu82}
M\"{u}ckenheim, W. (1982).
\newblock A resolution of the {Einstein}-{Podolsky}-{Rosen} paradox.
\newblock {\em Lettere al Nuovo Cimento}, 35(9):300--304.

\bibitem[M\"{u}ckenheim, 1986]{Mu86}
M\"{u}ckenheim, W. (1986).
\newblock A review of extended probabilities.
\newblock {\em Physics Reports}, 133(6):337--401.

\bibitem[Pitowsky, 1983]{Pi83}
Pitowsky, I. (1983).
\newblock Deterministic model of spin and statistics.
\newblock {\em Physical Review D}, 27(10):2316--2326.

\bibitem[Pitowsky, 1989]{Pi89}
Pitowsky, I. (1989).
\newblock {\em Quantum Probability---Quantum Logic}.
\newblock Springer-Verlag, New York.

\bibitem[Polubarinov, 1988]{Po88}
Polubarinov, I. (1988).
\newblock Continuous representation for spin 1/2, quantum probability and
  {Bell} paradox.
\newblock {\em Communication of the Joint Institute for Nuclear Research,
  Dubna}, E2-88-80.

\bibitem[Pusey et~al., 2012]{PuBaRu12}
Pusey, M., Barrett, J., and Rudolph, T. (2012).
\newblock On the reality of the quantum state.
\newblock {\em Nature Physics}, 8:475--478.

\bibitem[Rothman and Sudarshan, 2001]{RoSu01}
Rothman, T. and Sudarshan, E. (2001).
\newblock Hidden variables or positive probabilities?
\newblock {\em International Journal of Theoretical Physics}, 40(8):1525--1543.

\bibitem[Scully et~al., 1994]{ScWaSc94}
Scully, M., Walther, H., and Schleich, W. (1994).
\newblock {Feynman's} approach to negative probability in quantum mechanics.
\newblock {\em Physical Review A}, 49(3):1562--1566.

\bibitem[Sorkin, 1994]{So94}
Sorkin, R. (1994).
\newblock Quantum mechanics as quantum measure theory.
\newblock {\em Modern Physics Letters A}, 9(33):3119--3127.

\bibitem[Sorkin, 1997]{So97}
Sorkin, R. (1997).
\newblock Quantum measure theory and its interpretation.
\newblock In Feng, D. and Hu, B., editors, {\em Quantum Classical
  Correspondence: Proceedings of the 4th Drexel Symposium on Quantum
  Nonintegrability}, pages 229--251, Cambridge, Mass. International Press.

\bibitem[Srinivasan, 1995]{Sr95}
Srinivasan, S. (1995).
\newblock Complex measureable processes and path integrals.
\newblock In Sridhar, R., Srinivasa~Rao, K., and Lakshminarayanan, V., editors,
  {\em Selected Topics in Mathematical Physics: Professor R. Vasudevan Memorial
  Volume}, pages 13--27. Allied Publishers, New Delhi.

\bibitem[Srinivasan, 1998]{Sr98}
Srinivasan, S. (1998).
\newblock Complex measure, coherent state and squeezed state representation.
\newblock {\em Journal of Physics A}, 31:4541--4553.

\bibitem[Srinivasan, 2006]{Sr06}
Srinivasan, S. (2006).
\newblock Quantum phenomena via complex measure: Holomorphic extension.
\newblock {\em Fortschritte der Physik}, 54(7).

\bibitem[Srinivasan and Sudarshan, 1994]{SrSu94}
Srinivasan, S. and Sudarshan, E. (1994).
\newblock Complex measures and amplitudes, generalized stochastic processes and
  their application to quantum mechanics.
\newblock {\em Journal of Physics A}, 27:517--534.

\bibitem[Sudarshan and Rothman, 1993]{SuRo93}
Sudarshan, E. and Rothman, T. (1993).
\newblock A new interpretation of {Bell's} inequalities.
\newblock {\em International Journal of Theoretical Physics}, 32(7):1077--1086.

\bibitem[Suppes, 1965]{Su65}
Suppes, P. (1965).
\newblock Logics appropriate to empirical theories.
\newblock In {\em Symposium on the Theory of Models}. North Holland, Amsterdam.

\bibitem[Suppes, 1966]{Su66}
Suppes, P. (1966).
\newblock The probablistic argument for a nonclassical logic of quantum
  mechanics.
\newblock {\em Philosophy of Science}, 33:14--21.

\bibitem[Suppes and Zanotti, 1991]{SuZa91}
Suppes, P. and Zanotti, M. (1991).
\newblock Existence of hidden variables having only upper probabilities.
\newblock {\em Foundations of Physics}, 21(2):1479--1499.

\bibitem[Surya and Wallden, 2010]{SuWa10}
Surya, S. and Wallden, P. (2010).
\newblock Quantum covers in quantum measure theory.
\newblock {\em Foundations of Physics}, 40:585--606.

\bibitem[Van~Wesep, 2006a]{We06a}
Van~Wesep, R. (2006a).
\newblock Hidden variables in quantum mechanics: Generic models, set-theoretic
  forcing, and the appearance of probability.
\newblock {\em Annals of Physics}, 321:2453--2475.

\bibitem[Van~Wesep, 2006b]{We06b}
Van~Wesep, R. (2006b).
\newblock Hidden variables in quantum mechanics: Noncontextual generic models.
\newblock {\em Annals of Physics}, 321:2476--2490.

\bibitem[W\'{o}dkiewicz, 1988]{Wo88}
W\'{o}dkiewicz, K. (1988).
\newblock On the equivalence of nonlocality and nonpositivity of
  quasi-distributions in {EPR} correlations.
\newblock {\em Physics Letters A}, 121(1):1--3.

\bibitem[Youssef, 1991]{Yo91}
Youssef, S. (1991).
\newblock A reformulation of quantum mechanics.
\newblock {\em Modern Physics Letters A}, 6(3):225--235.

\bibitem[Youssef, 1994]{Yo94}
Youssef, S. (1994).
\newblock Quantum mechanics as {Bayesian} complex probability theory.
\newblock {\em Modern Physics Letters A}, 9(28):2571--2586.

\bibitem[Youssef, 1995]{Yo95}
Youssef, S. (1995).
\newblock Is complex probability theory consistent with {Bell's} theorem?
\newblock {\em Physics Letters A}, 204:181--187.

\bibitem[Youssef, 1996]{Yo96}
Youssef, S. (1996).
\newblock Quantum mechanics as an exotic probability theory.
\newblock In Hanson, K. and Silver, R., editors, {\em Maximum Entropy and
  Bayesian Methods}, pages 237--244. Kluwer, Dordrecht.

\end{thebibliography}

\end{document}